\documentclass[a4paper,11pt]{article}
\usepackage[english]{babel}
\usepackage[a4paper,left=3.5cm,right=3.5cm,top=3.0cm,bottom=3.0cm]{geometry}
\usepackage{amsfonts, amsmath, amssymb, amsthm, mathrsfs}
\usepackage{psfrag}
\usepackage{graphicx}
\usepackage[T1]{fontenc}
\usepackage[utf8]{inputenc}
\usepackage{authblk}
\usepackage{verbatim}

\providecommand{\abs}[1]{\lvert #1 \rvert}
\providecommand{\norm}[1]{\lVert #1 \rVert}
\providecommand{\trinorm}[1]{\lvert \mspace{-1.5mu}\lvert
\mspace{-1.5mu} \lvert #1 \rvert \mspace{-1.5mu}\rvert\mspace{-1.5mu}\rvert}

\providecommand{\norm}[1]{\lVert #1 \rVert}

\numberwithin{equation}{section}

\newtheorem{theorem}{Theorem}
\newtheorem{proposition}[theorem]{Proposition}

\title{\large{\textbf{SOME APPLICATIONS OF THE LEE-YANG THEOREM}}}
\date{}
\author[*]{J\" urg Fr\" ohlich}
\author[**]{Pierre-Fran\c cois Rodriguez}
\affil[*]{\textit{\small{Institute for Theoretical Physics, ETH Zurich, CH-8093, Zurich, Switzerland}}}
\affil[**]{\textit{\small{Department of Mathematics, ETH Zurich, CH-8092, Zurich, Switzerland}}}

\begin{document}
\maketitle

\vspace{-0.5cm}
\begin{quote}
\begin{center}
\textit{This note is dedicated to Elliott Lieb, mentor and friend, \\ on the occasion of his eightieth birthday.}
\end{center}
\end{quote}

\vspace{0.3cm}

\begin{abstract}
For lattice systems of statistical mechanics satisfying a Lee-Yang property (i.e., for which the Lee-Yang circle theorem holds), we present a simple proof of analyticity of (connected) correlations as functions of an external magnetic field $h$, for $\textrm{Re} \ \! h \neq 0$. A survey of models known to have the Lee-Yang property is given. We conclude by describing various applications of the aforementioned analyticity in $h$. 
\end{abstract}

\vspace{0.2cm}

\section{Introduction}

Approximately sixty years ago, while studying phase transitions in certain types of monatomic gases, Lee and Yang were led to investigate the location of zeroes of the (grand) partition function of such systems as a function of external parameters, in particular the chemical potential, as the thermodynamic limit is approached \cite{LY1}. At first sight, obtaining information on the location of these roots appears to be a formidable task. Yet, for a classical lattice gas with a variable chemical potential equivalent to the Ising model in an external magnetic field, they were able to show \cite{LY2} that the distribution of zeroes exhibits some astonishing regularity: all the roots lie on the imaginary axis (i.e., on the unit circle in the complex activity plane; see Theorem \ref{T1}, below). Their celebrated result was subsequently extended to plenty of further models.

In this note, we offer a short review of the present ``status quo'' regarding this so-called Lee-Yang theorem (Section 2), and then sketch some applications, using some novel arguments. Our main results can be found in Sections \ref{analyticity} and \ref{apps}. Here is a brief summary. For simplicity, we consider models on the lattice $\mathbb{Z}^d$. With each site, $x$, of the lattice we associate a random variable (a ``spin''), $\sigma_{x}$, taking values in a measure space $\Omega \subseteq \mathbb{R}^N$, for some $N \in \mathbb{N}$. The a-priori distribution of this random variable is specified by a probability measure, $\mu_0$, on $\Omega$. To an arbitrary  finite sublattice $\Lambda \subset \subset \mathbb{Z}^d$, there corresponds a space, \mbox{$\Omega^\Lambda = \big\{  \sigma_\Lambda := \{ \sigma_x \}_{x\in \Lambda} \; : \; \sigma_x \in \Omega, \ \forall x \in \Lambda   \big\}$}, of spin configurations in $\Lambda$. Interactions between the spins are described by a potential, $\Phi$, which associates with every $X \subset \subset \mathbb{Z}^d$ a continuous function $\Phi (X) \in C(\Omega^{X})$ representing the interaction energy of all the spins in $X$. Given a finite set $\Lambda \subset \subset \mathbb{Z}^d$, the Hamilton function or Hamiltonian (with free boundary conditions) of the system confined to $\Lambda$ is defined by
\begin{equation} \label{H_gen}
H^{\Phi}_\Lambda= \sum_{X \subset \Lambda} \Phi (X),
\end{equation}
and the corresponding partition function at inverse temperature $\beta >0 $ by
\begin{equation} \label{Z}
Z_{ \beta, \Lambda} ( \Phi ) = \int_{\Omega^\Lambda} e^{-\beta H^{\Phi}_\Lambda(\sigma_\Lambda) } \prod_{x \in \Lambda} d\mu_0( \sigma_x).
\end{equation}
By $\mathcal{B}$ we denote the ``large'' Banach space of interactions consisting of all translation-invariant potentials $\Phi$ satisfying $\trinorm{\Phi} := \sum_{X \ni 0} \abs{X}^{-1} \cdot \norm{\Phi(X)}_\infty < \infty$. 
It is a classical result (see for example \cite{Si1}, Theorem II.2.1, or \cite{Ru5}, Theorem 2.4.1) that, for $\Phi \in \mathcal{B}$ and arbitrary $\beta > 0$, the free energy density
\begin{equation} \label{TDfclass}
f(\beta, \Phi) := - \beta^{-1} \lim_{\Lambda \nearrow \mathbb{Z}^{d}} \abs{\Lambda}^{-1} \log Z_{ \beta, \Lambda} ( \Phi )
\end{equation}
exists and is finite, where the thermodynamic limit is understood in the sense of van Hove.

In what follows, the potential $\Phi$ always includes a contribution due to an external magnetic field, $h$. For example, if $N=1$ the Hamiltonian is of the form
\begin{equation*}
H^{\Phi}_\Lambda (\sigma_\Lambda)= H^{\Phi,0}_\Lambda(\sigma_\Lambda) - h\sum_{ x \in \Lambda} \sigma_x, \qquad \text{where } H^{\Phi,0}_\Lambda(\sigma_\Lambda) = H^{\Phi}_\Lambda(\sigma_\Lambda) \big\vert_{h=0}, \text{ } \sigma_{x}\in \mathbb{R}.
\end{equation*}
For certain choices of $(\Phi, \mu_0)$ (see Section \ref{LYsurvey} for an overview), the partition function $Z_{ \beta, \Lambda} ( \Phi )$ is known to be non-zero in the regions $\mathbb{H}_\pm = \{ h \in \mathbb{C}: \pm \mathrm{Re } \ \!  h >0 \}$. This is the famous Lee-Yang theorem. 

The main result of this note can be formulated as follows.
Assuming that the pair $( \Phi, \mu_0)$ and the boundary conditions imposed at $\partial {\Lambda}$ are such that the Lee-Yang theorem holds (and under suitable assumptions on the decay of the measure $\mu_0$ at infinity), 
\begin{equation} \label{MAINRESULT}
\begin{split}
&\text{the connected correlation functions $\langle \sigma_{x_1}  ;  \! \! \ \dots \!\! \ ;  \sigma_{x_n} \rangle^c_{ \Lambda, \beta, h}$ are analytic in $h$ and} \\
&\text{have a unique thermodynamic limit analytic in $h$ in the Lee-Yang regions $\mathbb{H}_\pm$},
\end{split}
\end{equation}
where $x_1, \dots, x_n$ are arbitrary sites in $\mathbb{Z}^{d}$ and $\beta >0$; (see \eqref{corrs} for the definition of connected correlations). The proof of this result is given in Section \ref{analyticity}; see, in particular, Proposition \ref{analytic_corrs}. 
At first, $\Lambda$ is chosen to be a rectangle and periodic boundary conditions are imposed at $\partial\Lambda$. But the result can be seen to hold for \textit{arbitrary} boundary conditions for which the Lee-Yang theorem is valid. 
Earlier results of this kind can be found in \cite{Le1, LePe1}, where the example of the Ising ferromagnet is treated. Our methods enable us to extend (\ref{MAINRESULT}) to systems of multi-component spins ($N \geq 2$) satisfying a Lee-Yang theorem, as discussed at the end of Section \ref{analyticity}. Such systems include the rotor and the classical Heisenberg model, with suitable ferromagnetic conditions imposed. Similar results (for ``Duhamel correlation functions'') can also be proven for certain quantum-mechanical lattice spin systems.

In Section \ref{apps}, we discuss some applications of (\ref{MAINRESULT}) to classical spin systems and scalar Euclidean field theories. Assuming that the magnetic field $h$ is different from 0, uniqueness of the thermodynamic limit of correlations for a large class of boundary conditions is proven, properties of the magnetization and of the correlation length are reviewed and some bounds on critical exponents for the magnetization and the correlation length as functions of the magnetic field are recalled.

\section{The Lee-Yang theorem - a tour d'horizon} \label{LYsurvey}

In this section, we present an overview of classical (and quantum) lattice systems for which a Lee-Yang theorem is known to hold. For the sake of clarity, the models of interest are divided into four groups.

\subsection*{Ising-type models}

\noindent These are one-component models ($N=1$) with a Hamiltonian $H_\Lambda$, $\Lambda \subset \subset \mathbb{Z}^d$, given by
\begin{equation} \label{His}
H_\Lambda (\sigma_\Lambda) = - \sum_{\{x,y \} \subset \Lambda} J_{xy} \sigma_x \sigma_y - \sum_{x \in \Lambda} h_x \sigma_x, 
\end{equation}
where the first sum is over all pairs in $\Lambda$ ,$J_{xx}=0$, for all $x$, and $J_{xy} = J_{yx}$, for all $x,y$. Among the best results concerning this class of models is one established by Newman \cite{New1}, which is summarized in the following theorem.

\begin{theorem} \label{T1}
If the pair interaction in \eqref{His} is ferromagnetic, i.e., the couplings $J_{xy}$ satisfy
\begin{equation} \label{IsFerConds}
J_{xy} = J_{yx} \geq 0,  \qquad \forall x,y,
\end{equation}
and $\mu_0$ is an arbitrary signed (i.e., real-valued) measure on $\mathbb{R}$ that is even or odd, has the property that $\int_{\mathbb{R}} e^{b \sigma^2} d|\mu_0(\sigma)| < \infty$, for all $b \geq 0$, and satisfies the condition
\begin{equation} \label{zerointsLYcond}
\hat{\mu}_0(h):= \int_{\mathbb{R}}e^{h\sigma}d\mu_0(\sigma) \neq 0, \qquad \forall h \in \mathbb{H}_+:= \{ z \in \mathbb{C}  :  \mathrm{Re } \ \! z>0 \},
\end{equation}
then, for arbitrary $\beta > 0$, the partition function corresponding to $H_\Lambda$ in \eqref{His} satisfies
 \begin{equation} \label{LY1} 
Z_{\beta, \Lambda} \big (\{h_x \}_{x \in\Lambda} \big)= \int e^{-\beta H_\Lambda(\sigma_\Lambda)} \prod_{x \in \Lambda} d\mu_{0}(\sigma_x) \neq 0 \qquad \mathrm{if} \ h_x \in \mathbb{H}_+, \ \forall x \in \Lambda.
\end{equation}
\end{theorem}
The condition \eqref{zerointsLYcond} is quite natural in that it requires \eqref{LY1} to hold for non-interacting spins, i.e., $J_{xy} =0$, for all $x,y$. (Clearly, \eqref{LY1} trivially implies \eqref{zerointsLYcond}.) Moreover, Theorem \ref{T1} continues to hold if, at different sites of the lattice, different a-priori distributions are chosen, provided each $\mu_{0,x}$, $x \in \Lambda$, satisfies the conditions on $\mu_0$ formulated above.

In the (physically most relevant) case of positive measures, Lieb and Sokal \cite{LiSo1} have obtained rather deep insights regarding the Lee-Yang property \eqref{LY1}. In order to summarize salient features of their findings, we introduce some further notation. For any $a \geq 0$ and $n \in \mathbb{N}$, let $\mathcal{A}_{a+}^n$ be the (Fr\'echet) space of entire functions on $\mathbb{C}^n$ satisfying $\norm{f}_b := \sup_{z \in \mathbb{C}^n} e^{-b |z|^2} |f(z)| < \infty$, for all $b> a$. Given an open set $O \subset \mathbb{C}^n$, we define $\mathcal{P} (O)$ to be the class of polynomials defined on $\mathbb{C}^n$ that do not vanish in $O$. 
We denote by $\overline{\mathcal{P}}_{a+} (O)$ its closure in $\mathcal{A}_{a+}^n$. Given an arbitrary distribution $\mu \in \mathcal{S}'(\mathbb{R}^{n})$ (the space of tempered distributions on $\mathbb{R}^{n}$), we say that $\mu \in T^{n}$ if, in addition, $ e^{a  |x|^2} \mu (x) \in \mathcal{S}'(\mathbb{R}^{n})$, for all $a >0$. A finite, positive measure $\mu$ on $\mathbb{R}^n$ is henceforth called a \textit{Lee-Yang measure} if $\mu \in T^n$ and if its Laplace transform, $\hat{\mu}$, satisfies $\hat{\mu} \in \overline{\mathcal{P}}_{0+} (\mathbb{H}^n_+)$, where $\mathbb{H}^n_+ = \mathbb{H}_+ \times \dots \times \mathbb{H}_+\subset \mathbb{C}^n$ and $\mathbb{H}_+ = \{ z \in \mathbb{C}  :  \mathrm{Re } \ \!  z>0 \}$, as before. 

\begin{theorem} \label{T2} \cite{LiSo1} If $\mu$ is a Lee-Yang measure on $\mathbb{R}^n$, $B \in \overline{\mathcal{P}}_{a+} (\mathbb{H}_+^n)$ for some $a \geq 0$, $B$ is non-negative on $\mathrm{supp}(\mu) \subseteq \mathbb{R}^n$ and strictly positive on a set of non-zero $\mu$-measure, then $B\mu$ is also a Lee-Yang measure on $\mathbb{R}^n$.
\end{theorem}

This generalizes Theorem \ref{T1} (for positive $\mu_{0}$). Indeed, first note that the product measure $d\mu (\sigma_\Lambda) = \prod_{x\in \Lambda} d \mu_0 (\sigma_x)$ with $\mu_0 \in T^1$ satisfying \eqref{zerointsLYcond} is a Lee-Yang measure on $\mathbb{R}^{|\Lambda|}$; see \cite[Corollary 3.3]{LiSo1}. Let $B$ be the Boltzmann factor corresponding to the Ising pair interaction in \eqref{His}, i.e., $B(\sigma_\Lambda)= \exp \big[ \sum_{x,y \in \Lambda} J_{xy} \sigma_x \sigma_y \big]$, with $J_{xy} \in \mathbb{R}$. It is not hard to see that $B \in \mathcal{A}_{\norm{J}+}^{|\Lambda|}$, where $\norm{J}$ refers to the matrix norm of $J = (J_{xy})_{x,y \in \Lambda}$ when $\mathbb{R}^{|\Lambda|}$ is equipped with the Euclidean norm. 
One then finds \cite[Proposition 2.7]{LiSo1} that $B \in \overline{\mathcal{P}}_{\norm{J}+} \big(\mathbb{H}^{|\Lambda|}_+\big)$ \textit{if and only if} $J_{xy} \geq 0$, for all $x,y \in \Lambda$. Under this condition it follows from Theorem \ref{T2} that the partition function $Z_{\beta, \Lambda}\big(\{h_x \}_{x \in\Lambda} \big)$ in Theorem \ref{T1} satisfies
\begin{equation} \label{T2cor}
Z_{\beta=1, \Lambda}\big(\{h_x \}_{x \in\Lambda} \big) = \widehat{B\mu}\big(\{h_x \}_{x \in\Lambda} \big) \in \overline{\mathcal{P}}_{0+} \big(\mathbb{H}_+^{|\Lambda|}\big),
\end{equation}
which, by virtue of Hurwitz' Theorem \cite[p. 178]{Ahlf1}, implies that \eqref{LY1} holds. Furthermore, \eqref{T2cor} says that $B$, the Boltzmann factor pertaining to the \textit{ferromagnetic} Ising pair interaction, is a ``multiplier'' for Lee-Yang measures.

The class of admissible measures $\mu$ in Theorem \ref{T2} may be further enlarged by specifying some $\textit{falloff}$ at $\infty$, see \cite[Definition 3.1]{LiSo1}, which amounts to requiring in Theorem \ref{T1} that $\int_{\mathbb{R}} e^{b \sigma^2} d|\mu_0(\sigma)| < \infty$, for \textit{some} $b > 0$ only. In this case, \eqref{LY1} only holds for sufficiently \textit{small} $\beta>0$, depending on the choice of $\lbrace J_{xy} \rbrace$; see \cite[Remark 1.1]{New1} and \cite[Corollary 3.3]{LiSo1}. 

On a historical note, the road from the seminal article of Lee and Yang \cite{LY2} on the Ising model (i.e. $\mu_0 =(\delta_1 + \delta_{-1})/2$ in \eqref{LY1}) to the treatise \cite{LiSo1} of Lieb and Sokal spanned almost three decades, with important contributions by Asano \cite{As2}, Suzuki \cite{Su1} and Griffiths \cite{Gr1} (all concerning discrete spins), Ruelle's proof of a more general zero theorem \cite{Ru1} for Ising spins (generalizing a contraction method first introduced by Asano \cite{As11}), and the work of Simon and Griffiths \cite{SiGr1}, which, among other things, establishes \eqref{LY1} for $d\mu_0 (\sigma) = \exp [-a \sigma^4 - b \sigma^2] d\sigma$, with $a>0$ and $b \in \mathbb{R}$. (This particular measure indeed satisfies \eqref{zerointsLYcond}, see \cite[Example 2.7]{New1}. It arises in the lattice approximation to the $(\phi^4)_{2,3}$ Euclidean field theory, cf. \cite{Si3}, Chapters VIII and IX, for which the Lee-Yang theorem is shown to hold \cite[Theorem 6]{SiGr1}.) 

\subsection*{One-component models with more general interactions}

\noindent Next, we consider Hamiltonians of the form
\begin{equation} \label{H2}
H^{\Phi}_\Lambda(\sigma_\Lambda) = \sum_{ X \subset \Lambda \; : \; \abs{X}\geq 2} \Phi(X) (\sigma_X) - \sum_{x \in \Lambda}h_x \sigma_x \ \equiv \  H^{\Phi,0}_\Lambda (\sigma_\Lambda)  - \sum_{x \in \Lambda}h_x (\sigma_x-1),
\end{equation}
for one-component spins $\sigma_x \in \mathbb{R}$, where we have added a constant linear in $\lbrace h_{x} \rbrace$ for later convenience. Theorem \ref{T2} implies that if the Boltzmann factor $B_\beta(\sigma_\Lambda)= \exp \big[- \beta H^{\Phi,0}_\Lambda (\sigma_\Lambda) \big]$, viewed as a function on $\mathbb{C}^{|\Lambda|}$ for fixed $\beta >0$, belongs to $\overline{\mathcal{P}}_{a+} \big(\mathbb{H}_+^{|\Lambda|}\big)$, for some $a \geq 0$, then the partition function $Z_{ \beta, \Lambda} \big (\{h_x \}_{x \in\Lambda} \big)$ corresponding to the Hamiltonian \eqref{H2} satisfies
 \begin{equation} \label{LY2} 
Z_{\beta, \Lambda}\big(\{h_x \}_{x \in\Lambda} \big) \neq 0, \qquad \mathrm{for} \ h_x \in \mathbb{H}_+, \text{ }x \in \Lambda,
 \end{equation}
for any Lee-Yang measure $\mu_0$ as defined above Theorem \ref{T2}. 

If $\mu_0 =(\delta_1 + \delta_{-1})/2$ (Ising spins), reasonably explicit results are known. The partition function is then seen to be a multi-affine polynomial in the activity variables $z_x = \exp[-2 \beta h_x]$, $x\in \Lambda$: 
\begin{equation} \label{Zgenints}
Z_{\beta, \Lambda}(\{z_x\}_{x \in \Lambda}) = \sum_{X \subset \Lambda} E_X(\beta) \prod_{x \in X} z_x, 
\end{equation}
where $E_X(\beta) = \exp \big[ - \beta H_\Lambda^{\Phi,0}(\sigma_\Lambda \ \vert \sigma_{x} = -1, x\in X, \; \sigma_{x} = 1, x\in \Lambda \setminus X) \big]$. The Lee-Yang property \eqref{LY2} then asserts that $Z_{\beta, \Lambda}(\{z_x\}_{x \in \Lambda})$ does not vanish whenever $|z_x|<1$, for all $x \in \Lambda.$ If $H_\Lambda^{\Phi,0}$ is invariant under $\sigma_x \mapsto -\sigma_x$, $x \in \Lambda$ (``spin-flip'' symmetry), the coefficients satisfy $E_X(\beta)= E_{\Lambda \setminus X}(\beta)$, for all $X \subset \Lambda$. In particular, if $z_x=z$, for all $x$, this yields $Z_{\beta, \Lambda}(z) = z^{|\Lambda|} \cdot Z_{\beta, \Lambda}(z^{-1})$, for all $z \neq 0$, and therefore $Z_{\beta, \Lambda} (z) = 0$ implies that $|z|=1$. This is the Lee-Yang circle theorem. Ruelle's recent results \cite[Lemma 8 and Theorem 9]{Ru4} yield the following theorem.
\begin{theorem} \label{T3}
Let $Z_{\beta, \Lambda}(\{z_x\}_{x \in \Lambda})$ be \textit{any} multi-affine polynomial of the form \eqref{Zgenints}, with coefficients $E_X(\beta) = E_{\Lambda \setminus X}(\beta) >0$, for all $X \subset \Lambda$, and all $\beta > 0$, that satisfies
\begin{equation} \label{LYmultiaffine}
Z_{\beta, \Lambda}(\{z_x\}_{x \in \Lambda}) \neq 0, \qquad \mathrm{if  } \ |z_x| < 1 \ \forall x \in \Lambda,
\end{equation}
for all $\beta >0$. Then $Z_{\beta, \Lambda}$ is the partition function obtained by choosing $\Phi$, in \eqref{H2}, to be an Ising ferromagnetic pair interaction, i.e., $\Phi(\{ x,y\})(\sigma_x, \sigma_y)= -J_{xy}\sigma_x \sigma_y$, for some $J_{xy}=J_{yx} \geq 0$, and $\Phi(X)=0$, otherwise.
\end{theorem}
This is a \textit{converse} to the Lee-Yang theorem for Ising spins. Ruelle also shows that, for general interactions $\Phi$ in \eqref{H2} with spin-flip symmetry, the Lee-Yang theorem can only hold at sufficiently \textit{low} temperatures. (We refer the reader to \cite{Ru4} for a precise statement, and to \cite{LeRu1} for a concrete example.) To further illustrate this point, we consider the Hamiltonian \eqref{H2} with
\begin{equation*}
H_\Lambda^{\Phi,0}(\sigma_\Lambda) := -  \sum_{U \subset \Lambda: \; |U|=2,4} J_U \sigma^U, \qquad \text{where } \sigma^U = \prod_{x \in U} \sigma_x \text{ and } J_U \in \mathbb{R},
\end{equation*} 
which has a pair and a four-spin interaction, and is invariant under $\sigma_x \mapsto -\sigma_x$, $x \in \Lambda$. The corresponding partition function $Z_{\beta, \Lambda}(z_\Lambda)$ is given by \eqref{Zgenints}. Using the identities $ \prod_{x \in U}(1 \pm \sigma_x) =  \sum_{X \subset U} (\pm 1)^{\abs{X}} \sigma^X$, for any $U \subset \Lambda$ (the term corresponding to $X = \emptyset$ is understood to be $1$), which imply
\begin{equation*}
p_U (\sigma_U) := \prod_{x \in U}\frac{1}{2}(1 + \sigma_x) +  \prod_{x \in U}\frac{1}{2}(1 - \sigma_x) \ = \ 2^{1-|U|} \sum_{\substack{X \subset U \\ |X| \text{ even}}} \sigma^X,
\end{equation*}
one deduces that $H_\Lambda^{\Phi,0}$ may be expressed in terms of $p_U (\sigma_U)$ rather than $\sigma^U$. Indeed, 
\begin{equation*} 
H_\Lambda^{\Phi,0} = -  \tilde{J}_0 \ - \sum_{ U \subset \Lambda : \; |U|=2,4} \tilde{J}_U  \cdot p_U (\sigma_U), \quad \text{where }  \tilde{J}_U =\left \{
\begin{array}{lr}
8 J_U, & \text{if } |U |=4  \\
2 J_U - 2 \sum_{ U \subset X \subset \Lambda: \; |X| =4} J_X, & \text{if } |U|=2 
\end{array} \right.,
\end{equation*}
and $\tilde{J}_0$ is an irrelevant additive constant, which we may neglect. We note that, if $\sigma_\Lambda$ is the unique spin configuration such that $\{ x \in \Lambda : \; \sigma_x = -1  \}=X$, for some given subset $X$ of $\Lambda$, the quantity $p_U(\sigma_U)$ vanishes unless $U \subset X$ or $U \subset \Lambda \backslash X$, in which case $p_U(\sigma_U)=1$. 
It follows that $Z_{ \beta, \Lambda} (z_\Lambda) = \sum_{X \subset \Lambda}  E_X^{(2)} (\beta) \cdot E_X^{(4)} (\beta) z^X \equiv \big( Z^{(2)}_{ \beta, \Lambda} * Z^{(4)}_{\beta, \Lambda}\big)(z_\Lambda)$ ($\ast$ is called \textit{Schur-Hadamard} product), where $Z^{(k)}_{ \beta, \Lambda}(z_\Lambda) = \sum_{X \subset \Lambda} E_X^{(k)} (\beta) z^X$, $k=2,4$, with coefficients $E_X^{(k)} (\beta) = \prod_{\substack{U \subset X \text{ or }  U \subset \Lambda \setminus X \\ |U|=k}} e^{\beta \tilde{J}_U}$, for all $X \subset \Lambda$. This $*$-product representation of $Z_{ \beta, \Lambda}$ is very useful, because it essentially allows one to consider $Z^{(k)}_{ \beta, \Lambda}$, $k=2,4$, separately, as we will now see. By \cite[Example 7(d)]{Ru4}, $Z^{(4)}_{ \beta, \Lambda}$ satisfies \eqref{LYmultiaffine} at temperature $\beta >0$ whenever 
\begin{equation}\label{condition1}
\beta \tilde{J}_U = 8 \beta J_U \geq \ln 2 \text{ (or $J_U =0$), for all $U \subset \Lambda$ with $|U|=4$},
\end{equation}
i.e., at sufficiently \textit{low} temperature, for given $J_U \geq 0$, $|U|=4$. Similarly, one obtains from \cite[Example 7(a)]{Ru4} that $Z^{(2)}_{ \beta, \Lambda}$ fulfills \eqref{LYmultiaffine} if $\beta \tilde{J}_U \geq 0$, for all $U \subset \Lambda$ with $|U|=2$, or, equivalently, if
\begin{equation}\label{condition2}
J_U \geq  \sum_{X: \; U \subset X \subset \Lambda \text{ and } |X| =4} J_X, \quad \text{for all $U \subset \Lambda$ with $|U|=2$},
\end{equation}
which implies that $J_U \geq 0$, for $|U|=2$, and also requires the four-spin interaction to be suitably \textit{small} as compared to the pair interaction. Finally, since $ Z_{ \beta, \Lambda} = Z^{(2)}_{ \beta, \Lambda} * Z^{(4)}_{\beta, \Lambda}$, Proposition 2(c) in \cite{Ru4} (see also  \cite[Corollary 2.15]{Rod1}; a crucial ingredient of the proof is a contraction method first introduced by Asano \cite[Definition 3]{As11}) yields that $Z_{ \beta, \Lambda}$, subject to the constraints \eqref{condition1} and \eqref{condition2}, satisfies the Lee-Yang property \eqref{LYmultiaffine}, which is consistent with Theorem \ref{T3}.

\subsection*{Multi-component spins}

\noindent For models of classical $N$-component spins, $N \geq 2$, we denote the spin variable at site $x$ by $\boldsymbol{\sigma}_x = ( \sigma_x^i \; : \; i=1,\dots,N)$, (we use Latin superscripts to indicate the components). For $N=1$, the ``Lee-Yang region'' $\mathbb{H}_+$ is the region where the Laplace transform of the single-spin measure does not vanish, cf. condition \eqref{zerointsLYcond}. An example of a suitable generalization to $N\geq 2$ is to consider measures $\mu \in T^N$ (see above Theorem \ref{T2}) obeying 
\begin{equation} \label{NzerointsLYcond}
\int e^{h \sigma^1} d \mu(\boldsymbol{\sigma}) \neq 0, \qquad \text{for} \quad \mathrm{Re} \ \! h \neq 0,
\end{equation}
i.e., the magnetic field $h$ is assumed to point in the 1-direction. If, in addition to satisfying \eqref{NzerointsLYcond}, $\mu$ is assumed to be rotationally invariant, then it follows \cite[Proposition 4.1]{LiSo1} that its Laplace transform $\hat{\mu}( \boldsymbol{h})$, $\boldsymbol{h} = (h^1, \dots, h^N)$, belongs to $\overline{\mathcal{P}}_{0+}(\Omega_N)$, where $\Omega_N = \Omega_N^+ \cup \Omega_N^-$ and
\begin{equation} \label{Zero_free_domain}
 \Omega_N^\pm = \Big\{ \boldsymbol{h}=(h^1, \dots, h^N) \in \mathbb{C}^N \; : \; \pm \mathrm{Re} \ \! h^1 >  \sum_{i=2}^N |h^i| \Big\}.
\end{equation}
For the Hamiltonian
\begin{equation} \label{H3}
H_\Lambda = - \sum_{ \{ x , y \} \subset \Lambda} \sum_{i=1}^N J_{xy}^i \sigma_x^i \sigma_y^i \ -  \ \sum_{x\in \Lambda} \sum_{i=1}^N h_x^i\sigma_x^i, 
\end{equation}
the following result \cite[Corollaries 4.4 and 5.5]{LiSo1} (see also \cite{Dun1} for the plane rotator model) is known. 
\begin{theorem} \label{T4}
Let $J_{xy}^i\in \mathbb{R}$, for all $x,y \in \Lambda$ and $i=1,\dots,N$, with 
\begin{equation} \label{NBadConds}
J_{xy}^1 \geq \sum_{i=2}^N |J_{xy}^i|, \qquad \forall x,y \in \Lambda,
\end{equation}
and assume that $\mu_0 \in T^N$ is a rotationally invariant measure on $\mathbb{R}^N$ satisfying condition \eqref{NzerointsLYcond}. Then the partition function $Z_{\beta, \Lambda}\big(\{\boldsymbol{h}_x \}_{x \in\Lambda} \big)$, $\beta >0$, corresponding to the Hamiltonian \eqref{H3} does not vanish whenever $\boldsymbol{h}_x \in \Omega_N^+$, for all $x \in \Lambda$.
\end{theorem}
The ``Lee-Yang'' region $\Omega_N^+$ in Theorem \ref{T4}, where the partition function $Z_{\beta, \Lambda}\big(\{\boldsymbol{h}_x \}_{x \in\Lambda} \big)$ is non-zero, can usually be further enlarged; see for example \cite[Proposition 4.1]{LiSo1}. The domain \eqref{Zero_free_domain} considered here will suffice for our purposes. In particular, the possibility to include small but non-zero transverse fields $h_x^i$, $i=2,\dots,N$, $x \in \Lambda$, will prove useful in what follows; see Section \ref{analyticity}.

For plane rotator models, i.e., $N=2$, Theorem \ref{T4} is ``optimal'' in the sense that $J_{xy}^1 \geq |J_{xy}^2|$ is a sensible generalization of \eqref{IsFerConds}. For $N \geq 3$, the above result is not entirely satisfactory. Indeed, the constraint \eqref{NBadConds} is a ferromagnetic condition that forces the interaction to be anisotropic and hence entails an explicit breaking of the $O(N)$-symmetry. Ideally, one would like to replace that condition by the more natural constraint
\begin{equation} \label{NGoodConds}
J_{xy}^1 \geq \max_{2 \leq i \leq N} |J_{xy}^i|.
\end{equation} 
This is possible for the classical \textit{Heisenberg} model ($N=3$ and $\mu_0$ the uniform distribution on the unit sphere). See Proposition \ref{PclassHeis} below. 

\subsection*{Quantum spins}

\noindent Quantum lattice systems for which the Lee-Yang theorem is known to hold include rather general (anisotropic) Heisenberg models with suitable ferromagnetic pair interactions. Although we are primarily concerned with models of classical spins, we quote some results concerning such models for the sake of completeness and because they imply the improved result for the classical Heisenberg model mentioned above. We fix a region $\Lambda \subset \subset \mathbb{Z}^d$ and a spin $s \in \mathbb{N}/2$. By $\hat{\sigma}_x^i$, $x \in \Lambda$, $i=1,2,3$, we denote the $i$-th component of the quantum-mechanical spin operator at site $x$, in the spin-$s$ representation of $su(2)$. The ``hat'' distinguishes these operators from classical spins. They act on the space $ (\mathbb{C}^{2s+1})^{\otimes \vert\Lambda \vert}$ and satisfy the usual commutation relations $[\hat{\sigma}_x^j,\hat{\sigma}_y^k] = i \delta_{x,y} \epsilon_{jkl}\hat{\sigma}_x^l $ of generators of $su(2)$. We consider the Hamiltonian
\begin{equation} \label{H4}
\widehat{H}_{\Lambda,s} = - \sum_{ \{ x , y \} \subset \Lambda} \sum_{i=1}^3 J_{xy}^i \hat{\sigma}_x^i \hat{\sigma}_y^i \ -  \ \sum_{x\in \Lambda} \sum_{i=1}^3 h_x^i\hat{\sigma}_x^i,
\end{equation}
for arbitrary $s \in \mathbb{N}/2$, and define the partition function by
\begin{equation} \label{Zquant}
Q_{\beta,\Lambda,s} \big( \{  \boldsymbol{h}_x \}_{x \in \Lambda}  \big) = (2s+1)^{-|\Lambda|} \cdot \mathrm{tr} \big[ \exp \big(- \beta \widehat{H}_{\Lambda,s}( \{ \boldsymbol{h}_x \}_{x \in \Lambda}) \big) \big], \qquad \beta >0.
\end{equation}
The following theorem is essentially due to Asano \cite{As11} (see also \cite{SuFi1}) and \cite[Theorem 3]{Dun1}, where \textit{complex} transverse magnetic fields are considered. 

\begin{theorem} \label{T5}
Assume that the couplings $J_{xy}^i$ in \eqref{H4} satisfy the ferromagnetic conditions $J_{xy}^1 \geq |J_{xy}^2| $ and $J_{xy}^1 \geq |J_{xy}^3|$, for all $x,y \in \Lambda$. Then, for arbitrary $s \in \mathbb{N}/2$ and $\beta >0$, the partition function $Q_{\beta,\Lambda,s} \big( \{  \boldsymbol{h}_x \}_{x \in \Lambda}  \big)$ in \eqref{Zquant} is non-zero whenever $  \boldsymbol{h}_x \in \Omega_3^+$ $($see \eqref{Zero_free_domain}$)$, for all $x \in \Lambda$.
\end{theorem}

For results somewhat more general than Theorem \ref{T5} see the references given above. Our formulation is tailored to our purposes. In particular, the following result is a corollary of Theorem \ref{T5}.

\begin{proposition} \label{PclassHeis}
Let $H_\Lambda$ be the Hamilton function defined in \eqref{H3}, with $N=3$, and let $\mu_0$ be the uniform measure on the unit sphere $S^2$. If the ferromagnetic conditions $J_{xy}^1 \geq |J_{xy}^2|$ and $J_{xy}^1 \geq |J_{xy}^3|$, for all $x,y \in \Lambda$, are satisfied, the partition function $Z_{\beta, \Lambda} ( \{ \boldsymbol{h}_x \}_{x \in \Lambda})$, $\beta >0$, does not vanish whenever $\boldsymbol{h}_x \in \Omega_3^+$, for all $x \in \Lambda.$
\end{proposition}

\begin{proof}
We define $Q_{\beta,\Lambda,s}^{(resc)} \big( \{ \boldsymbol{h}_x \}_{x \in \Lambda}  \big)$ to be the partition function corresponding to the Hamiltonian $\widehat{H}_{\Lambda,s}^{(resc)}$ defined as in \eqref{H4}, but with rescaled coefficients, $J_{xy}^i \mapsto J_{xy}^i/s^2$ and $h_x^i \mapsto h_x^i/s$. From a general result on classical limits of quantum spin systems due to Lieb \cite{Li1}, it follows that 
\begin{equation}\label{Lieb_class_limit}
\lim_{s \to \infty}Q_{\beta,\Lambda,s}^{(resc)} \big( \{ \boldsymbol{h}_x \}_{x \in \Lambda}  \big) = Z_{\beta, \Lambda} \big( \{ \boldsymbol{h}_x \}_{x \in \Lambda}\big), \text{ pointwise, for real $h_x^i$, $i=1,2,3$ and $x \in \Lambda$}.
\end{equation}
 Moreover, since $\big| Q_{\beta,\Lambda,s}^{(resc)} \big( \{ \boldsymbol{h}_x \}_{x \in \Lambda}  \big) \big|  \leq \exp \big( \beta \norm{\widehat{H}_{\Lambda,s}^{(resc)}} \big)$ (where $\norm{\cdot}$ denotes the operator norm) and $\norm{\widehat{H}_{\Lambda,s}^{(resc)}}$ is bounded uniformly in $s$, one sees that, for arbitrary compact subsets $K \subset \mathbb{C}^{3|\Lambda|}$,
\begin{equation} \label{Q_unif_bounds} 
\sup_{s \in \mathbb{N}/2; \ \{ \mathbf{h}_x \}_{x \in \Lambda} \in K} \big| Q_{\beta,\Lambda,s}^{(resc)} \big( \{ \boldsymbol{h}_x \}_{x \in \Lambda}  \big)\big| < \infty. 
\end{equation}
Assuming that $\Lambda=\{x_1,\dots, x_n\}$, we write $\boldsymbol{h_i} \equiv \boldsymbol{h_{x_i}}$ and consider the hypothesis
\begin{align*}
(H_k): \quad 
&\text{$ \lim_{s \to \infty }Q^{(resc)}_{\beta,\Lambda,s} \big( \boldsymbol{h}_{k} ; \; \{ \boldsymbol{h}_j \}_{j \neq k} \big)=Z_{\beta,\Lambda} \big( \boldsymbol{h}_{k} ; \; \{ \boldsymbol{h}_j \}_{j \neq k} \big)$, pointwise, and  $Z_{\beta,\Lambda} \big( \boldsymbol{h}_{k} ; \; \{ \boldsymbol{h}_j \}_{j \neq k} \big) \neq 0$,} \\
&\text{ whenever $\boldsymbol{h}_k \in \Omega_3^+$, for arbitrary (fixed) $\boldsymbol{h}_j \in \Omega_3^+$, $j<k$, and $\boldsymbol{h}_j \in \mathbb{R}^3$, $j>k$},
\end{align*}
for $k=1,\dots,n$. We begin by showing $(H_1)$. To this end, we first define the sequence of functions $\big(f^1_{s,1} \big)_s$ with $f_{s,1}^1(h_1^1)=Q_{\beta,\Lambda,s}^{(resc)} \big( h_1^1 ; \; h_1^2, h_1^3, \{ \boldsymbol{h}_j \}_{j \neq 1} \big)$, for arbitrary (fixed) $h_{1}^2,h_{1}^3 \in \mathbb{R}$ and $\boldsymbol{h}_j \in \mathbb{R}^3$, $j>1$. By \eqref{Lieb_class_limit}, $f_{s,1}^1$ converges pointwise for real $h_1^1$ as $s \to \infty$, and by \eqref{Q_unif_bounds}, it is uniformly bounded on compact subsets of $\mathbb{C}$. It thus follows by Vitali's Theorem \cite[Section 5.21]{Titch1} that $\lim_{s \to \infty} f_{s,1}^1(h_1^1)=Z_{\beta,\Lambda} \big( h_1^1 ; \; h_1^2, h_1^3, \{ \boldsymbol{h}_j \}_{j \neq 1} \big)$, uniformly on compact subsets of $\mathbb{C}$. Moreover, since $f_{s,1}^1$ does not vanish in $\mathbb{H}_+$ by Theorem \ref{T5}, Hurwitz' Theorem \cite[p. 178]{Ahlf1} implies that $Z_{\beta,\Lambda} \big( h_1^1 ; \; h_1^2, h_1^3, \{ \boldsymbol{h}_j \}_{j \neq 1} \big) \neq 0$ for $h_1^1 \in \mathbb{H}_+$.

Next, we consider the functions $f_{s,1}^2(h_1^2)=Q_{\beta,\Lambda,s}^{(resc)} \big( h_1^2 ; \; h_1^1, h_1^3, \{ \boldsymbol{h}_j \}_{j \neq 1} \big)$, for arbitrary (fixed) $h_1^1 \in \mathbb{H}_+$, $h_{1}^3 \in \mathbb{R}$ and $\boldsymbol{h}_j \in \mathbb{R}^3$, $j>1$. By what we have just shown, $f_{s,1}^2 (h_1^2)$ converges as $s \to \infty$ for all $h_1^2 \in \mathbb{R}$ and the limit is non-zero. Hence, Vitali's Theorem, together with the bounds \eqref{Q_unif_bounds}, yields that $\lim_{s \to \infty} f_{s,1}^2(h_1^2)=Z_{\beta,\Lambda} \big( h_1^2 ; \; h_1^1, h_1^3, \{ \boldsymbol{h}_j \}_{j \neq 1} \big)$, uniformly on compact subsets of $\mathbb{C}$; and, using Hurwitz' Theorem together with Theorem \ref{T5}, one shows that this limit does not vanish for $h_1^2 \in \{z \in \mathbb{C} : \; |z| < \mathrm{Re} \ \! h_1^1 \}$. Repeating this argument for the sequence $f_{s,1}^3(h_1^3) =  Q_{\beta,\Lambda,s}^{(resc)} \big( h_1^3 ; \; h_1^1, h_1^2, \{ \boldsymbol{h}_j \}_{j \neq 1} \big)$, with arbitrary $h_1^1 \in \mathbb{H}_+$, $h_{1}^2 \in \{z \in \mathbb{C} : \; |z| < \mathrm{Re} \ \! h_1^1 \}$ and $\boldsymbol{h}_j \in \mathbb{R}^3$, $j>1$, one obtains $(H_1)$.

In order to complete the proof of Proposition \ref{PclassHeis}, it suffices to show that $(H_{k-1})$ implies $(H_{k})$, $k=2,\dots,n$. The proof is completely analogous to the one of $(H_1)$ (one considers subsequently three sequences of functions $\big(f_{s,k}^i \big)_s$, for $i=1,2,3$). The only difference is that one invokes $(H_{k-1})$ instead of \eqref{Lieb_class_limit} to argue that the functions $f_{s,k}^1(h_k^1)=Q_{\beta,\Lambda,s}^{(resc)} \big( h_k^1 ; \; h_k^2, h_1^3, \{ \boldsymbol{h}_j \}_{j \neq k} \big)$, (for fixed $\boldsymbol{h}_j \in \Omega_3^+$, $j<k$, $h_k^2,h_k^3 \in \mathbb{R}$, and $\boldsymbol{h}_j \in \mathbb{R}^3$, $j>k$), converge pointwise, for \textit{real} $h_k^1$,  towards a non-zero limit, as $s \to \infty$.
\end{proof}

\section{Analyticity of correlations} \label{analyticity}

Our aim in this section is to establish analyticity of (connected) correlations as functions of an external magnetic field $h$ in the Lee-Yang regions $\mathbb{H}_\pm = \{ h \in \mathbb{C}  :  \pm \mathrm{Re} \ \! h > 0\}$, for the models introduced in the last section. For the sake of clarity, we first focus on models of Ising spins. Extensions to other models satisfying the Lee-Yang theorem are outlined at the end of this section. Thus, we consider the Ising Hamiltonian $H_\Lambda$ defined in \eqref{His}, but with a homogenous magnetic field $h_x = h$ for all $x$, and we write $H^{0}_\Lambda= {H_\Lambda}{\big\vert_{h=0}}$. For the purposes of this section, we further impose \textit{periodic} boundary conditions and require the pair interaction couplings to satisfy
\begin{equation}\label{Jconds}
J_{xy}=J_{yx}=J_{x-y,0} \geq 0, \quad \text{for all } x \neq y, \quad \text{and} \quad \sum_{x\neq 0}J_{0x}< \infty.
\end{equation}
The a-priori measure $\mu_0$ on $\mathbb{R}$ is assumed to be any even measure with \textit{compact} support satisfying condition \eqref{zerointsLYcond} (e.g., $\mu_0 =(\delta_{-1}+\delta_1)/2$, the example of Ising spins, or $\mu_0 = \lambda \big\vert_{[-1,1]} /2$, the example of a continuous spin uniformly distributed on $[-1,1]$.) The partition function at inverse temperature $\beta > 0$ is denoted by $Z_{ \beta, \Lambda}(h)$.
Theorem \ref{T1} holds for this class of models.

We are interested in  the \textit{connected} correlation functions $\langle \sigma_{x_1}  ;  \! \! \  \dots \! \! \  ;  \sigma_{x_n} \rangle^c_{\Lambda, \beta, h}$ (also called \textit{Ursell functions}, or \textit{cumulants} in a probability theory context), where $n \in \mathbb{N}$ and $\{x_1, \dots, x_n \} \subset \Lambda$, which can be defined as \cite[Section II.12]{Si1}
\begin{equation} \label{corrs}
\langle \sigma_{x_1}  ; \! \! \  \dots \! \! \  ;  \sigma_{x_n} \rangle^c_{\Lambda, \beta, h} = \Big[ \frac{\partial^n}{\partial \varepsilon_1 \cdots \partial \varepsilon_n} \log \big( \big\langle \exp \big[ \sum_{i=1}^n \varepsilon_i \sigma_{x_i} \big] \big\rangle_{\Lambda, \beta, h} \big) \Big] \bigg\vert_{\varepsilon_1= \cdots = \varepsilon_n =0},
\end{equation}
where $\langle \ \cdot \ \rangle_{\Lambda, \beta, h}$ denotes a (finite-volume) thermal average, i.e., an expectation with respect to the probability measure $\big(Z_{ \beta, \Lambda}(h)\big)^{-1} e^{-\beta H_\Lambda(\sigma_\Lambda)} \prod_{x \in \Lambda} d\mu_0(\sigma_x)$. Our aim is to prove the following result.

\begin{proposition} \label{analytic_corrs}
Under the above assumptions, for arbitrary sites $ x_1,\dots , x_n$, $n \in \mathbb{N}$, and for all $\beta>0$, the  thermodynamic limit 
\begin{equation*}
\langle \sigma_{x_1}  ;  \! \! \ \dots \!\! \ ;  \sigma_{x_n} \rangle^c_{ \beta, h} := \lim_{\Lambda \nearrow \mathbb{Z}^{d}} \langle \sigma_{x_1}  ;   \! \! \ \dots  ;  \! \! \ \sigma_{x_n} \rangle^c_{\Lambda, \beta, h}
\end{equation*}
of the connected correlation $\langle \sigma_{x_1}  ;   \! \! \ \dots  \! \! \ ;  \sigma_{x_n} \rangle^c_{\Lambda, \beta, h}$ exists and is an analytic function of the magnetic field $h$, for $\mathrm{Re} \ \! h \neq 0$. 
\end{proposition}

\begin{proof}
Since $\beta >0$ can be absorbed into a redefinition of $J_{xy}$ and $h$, there is no loss of generality in setting $\beta=1$. We consider the modified partition function
\begin{equation} \label{modZ}
Z_{ \Lambda} (h,\varepsilon_1, \dots, \varepsilon_n ) := \int \exp \Big\{ -  H^{0}_{\Lambda}(\sigma_\Lambda) + \sum_{x \in \Lambda} \Big[ h + \sum_{\alpha=1}^n \varepsilon_\alpha e^{ik_\alpha \cdot x} \Big] \sigma_x \Big\} \prod_{x \in \Lambda}d \mu_0(\sigma_x), 
\end{equation}
with $\boldsymbol\varepsilon=(\varepsilon_1, \dots, \varepsilon_n)\in \mathbb{C}^n$ and $k_\alpha \in \Lambda^*$, for all $\alpha=1,\dots,n$, where $\Lambda^*$ denotes the dual lattice of $\Lambda$. Assuming that $\mathrm{Re} ( h ) >0$, setting $h_x := h + \sum_{\alpha=1}^n \varepsilon_\alpha e^{ik_\alpha \cdot x}$ and noting that $\mathrm{Re} ( h_x ) \geq \mathrm{Re} ( h ) - \sum_{\alpha=1}^n | \varepsilon_\alpha |$, Theorem \ref{T1} is seen to have  the following corollary:
\begin{equation} \label{LYcor}
Z_{ \Lambda} (h,\boldsymbol{\varepsilon}  ) \neq 0 \qquad \text{if} \qquad \mathrm{Re} (h) > \sum_{\alpha=1}^n | \varepsilon_\alpha |.
\end{equation}
More generally, $Z_{ \Lambda} (h,\boldsymbol{\varepsilon})$ does not vanish whenever $|\mathrm{Re} ( h ) | > \sum_{\alpha=1}^n | \varepsilon_\alpha |$, because $Z_{ \Lambda}(\{h_x \}_{x \in \Lambda})=Z_{ \Lambda}(\{-h_x \}_{x \in \Lambda})$, by symmetry, since $\mu_0$ is even. In the following, we may therefore assume that $\mathrm{Re} ( h ) >0$. Defining the convex domain
\begin{equation} \label{domain_D}
D= \Big\{(h,\boldsymbol\varepsilon) \in \mathbb{C}^{n+1} \ : \ \mathrm{Re} ( h ) > \sum_{\alpha=1}^n | \varepsilon_\alpha |   \Big\} \subset \mathbb{C}^{n+1} ,
\end{equation}
it is an easy exercise to check that $Z_{ \Lambda} (h,\boldsymbol{\varepsilon} )$ has an \emph{analytic} logarithm in $D$, so that 
\begin{equation} \label{fdef}
f_{ \Lambda} (h,\boldsymbol{\varepsilon}  ) :=|\Lambda|^{-1} \log Z_{ \Lambda} (h,\boldsymbol{\varepsilon}  )
\end{equation}
is a well-defined, analytic function of $(h,\boldsymbol{\varepsilon} )$ in $D$; (we choose a determination of $\log$ that is real, for $h>0$ and $\boldsymbol{\varepsilon}=0$). 
Henceforth, we fix some wave vectors $k_1, \dots, k_n$ and consider a family of rectangular domains, $\Lambda$, with the property that $\Lambda^{*} \ni k_{\alpha}$, $\forall \alpha$. We then consider the analytic function
\begin{equation} \label{defR}
\mathcal{R}_{ \Lambda} (h,\boldsymbol{\varepsilon}  ) :=  Z_{ \Lambda} (h,\boldsymbol{\varepsilon}  )^{\frac{1}{|\Lambda|}}:= \exp \big[f_{ \Lambda} (h,\boldsymbol{\varepsilon}  ) \big], \qquad \text{for $(h,\boldsymbol{\varepsilon} )\in D$.}
\end{equation}
We note that $\mathcal{R}_{ \Lambda}$ is uniformly bounded in $\Lambda$ on arbitrary compact subsets of $D$: abbreviating $ - \sum_{x \in \Lambda} \big[ h + \sum_{\alpha=1}^n \varepsilon_\alpha e^{ik_\alpha \cdot x} \big] \sigma_x $ by $H_{\Lambda}^{h, \boldsymbol\varepsilon} (\sigma_\Lambda)$, definition \eqref{modZ} implies that, for all $(h, \boldsymbol{\varepsilon}) \in D$,
\begin{align*}
|\mathcal{R}_{ \Lambda} (h,\boldsymbol{\varepsilon}  ) |= |Z_{ \Lambda} (h,\boldsymbol{\varepsilon}  ) |^{1/|\Lambda|}  \leq \big(\norm{ e^{- (H_\Lambda^0 + H_{\Lambda}^{h, \boldsymbol\varepsilon})}}_\infty\big)^{1/|\Lambda|} \ \leq \   \exp \big[  |\Lambda|^{-1} ( \norm{H^{0}_{\Lambda}}_\infty+ \norm{  H_{\Lambda}^{h, \boldsymbol\varepsilon}}_\infty) \big].
\end{align*}
Clearly,
\begin{equation*}
\norm{ H_{\Lambda}^{h, \boldsymbol\varepsilon}}_\infty  \leq c \abs{\Lambda} \cdot \Big( |h| +  \sum_{\alpha =1}^n \abs{\varepsilon_\alpha} \Big), 
\end{equation*}
for some constant $c>0$; (here, we assume for simplicity that the support of the measure $\mu_0$ is compact).  Moreover, denoting by $\Phi$ the ferromagnetic pair interaction, i.e., $\Phi(\{x,y \})= - J_{xy}\sigma_x \sigma_y$, for all $x \neq y$, $\Phi(X)=0$ else, we have (see for example \cite[Section II.3]{Si1}, and recall the definition of $\trinorm{\cdot}$ below \eqref{Z}) that
\begin{align*}
\norm{H_\Lambda^0}_\infty \leq c' \abs{\Lambda} \cdot \trinorm{\Phi},
\end{align*}
where $c'>0$ may depend on the dimension $d$ of the lattice and takes into account the periodic boundary conditions imposed at $\partial \Lambda$. Notice that this last estimate holds for arbitrary translation-invariant interactions $\Phi$. The last two estimates imply that $\abs{\mathcal{R}_{ \Lambda} (h,\boldsymbol{\varepsilon}  ))} \leq c(h,\boldsymbol{\varepsilon})$, where $c(h,\boldsymbol{\varepsilon}  )$ depends continuously on $(h, \boldsymbol{\varepsilon})$ and does not depend on $\Lambda$. Thus, for any compact $K \subset D$, 
\begin{equation} \label{Runifbounds}
\sup_{\Lambda \ \!\! ; \ \!\! (h,\boldsymbol{\varepsilon})\in K} \abs{\mathcal{R}_{ \Lambda} (h,\boldsymbol{\varepsilon}  )} < \infty.
\end{equation} 

With uniform bounds at hand, we may study the thermodynamic limit $\Lambda \nearrow \mathbb{Z}^{d}$. With the help of a cluster expansion (see for example \cite{Uel1}) at large magnetic fields, one shows the existence of the limit
\begin{equation} \label{TDf}
f_{\infty}(h, \boldsymbol{\varepsilon}) = \lim_{\Lambda \nearrow \mathbb{Z}^{d}} f_{\Lambda} (h,\boldsymbol{\varepsilon})
\end{equation}
at any point in the (large-field) regime
\begin{equation} \label{largefieldconds}
\mathrm{Re} ( h )  > h_0, \qquad \abs{\mathrm{Im}(h)} < \delta, \qquad \abs{\varepsilon_\alpha} < \delta, \quad \forall \alpha=1, \dots, n,
\end{equation}
for sufficiently large $h_0$ and small $\delta >0$; (had we kept the $\beta$-dependence explicit, $h_0$ and $\delta$ would depend on $\beta$). Similar conclusions hold for $\mathcal{R}_{\Lambda}$. Let $S$ denote the subregion of $D$ determined by the constraints \eqref{largefieldconds}. At any point $(h, \boldsymbol{\varepsilon}) \in S$, the limit $\mathcal{R}_{\infty}(h, \boldsymbol{\varepsilon}) = \lim_{\Lambda \nearrow \mathbb{Z}^{d}}\mathcal{R}_{\Lambda}(h, \boldsymbol{\varepsilon})$ exists and is finite, and $\mathcal{R}_{\infty}(h, \boldsymbol{\varepsilon})  = \exp \big[f_{\infty}(h, \boldsymbol{\varepsilon})  \big]$, by \eqref{defR}. In particular, $\mathcal{R}_{\infty} \neq 0$ on $S$. Moreover, since the subregion $S \subset D$ is a determining set for $D$, and since the functions $\mathcal{R}_{\Lambda}$ are uniformly bounded on compact subsets $K$ of $D$, by \eqref{Runifbounds}, it follows from Vitali's Theorem \cite[Section 5.21]{Titch1} that $\mathcal{R}_{\Lambda}$ converges everywhere in $D$, as $\Lambda \nearrow \mathbb{Z}^{d}$, uniformly on any such $K$, and from Weierstrass' Theorem \cite[p. 176]{Ahlf1} that the limit $\mathcal{R}_{\infty}$ is \textit{analytic} in $D$. 

Next, we show that $\mathcal{R}_{\infty}$ vanishes nowhere on $D$. To this end, we first consider the functions $\mathcal{R}_{\Lambda}(h,\boldsymbol{0}) = \exp [f_{ \Lambda}(h,\boldsymbol{0})]$; $f_{ \Lambda}(h,\boldsymbol{0})$ is known to have a (pointwise) limit for \textit{real} $h$ \cite[Section II.3]{Si1}. 
Hence $\mathcal{R}_{\infty} (h,\boldsymbol{0}) \neq 0$, for arbitrary real $h>0$. 
It follows from Hurwitz' Theorem \cite[p. 178]{Ahlf1} that $\mathcal{R}_{\Lambda}(h,\boldsymbol{0}) \neq 0$, for all $h \in \mathbb{H}_+$. Let $(h, \boldsymbol{\varepsilon})$ be any point in $D$. We show that $\mathcal{R}_{ \infty}(h,\boldsymbol{\varepsilon}) \neq 0$ by repeated application of Hurwitz' Theorem: 
Define $g_\Lambda^1(z)=\mathcal{R}_{ \Lambda} (h,z,0, \dots, 0)$, which is analytic and converges uniformly on compact subsets of $D_1 := \{z \in \mathbb{C}  :   \abs{z} < \mathrm{Re} ( h )  \}$ towards $g_\infty^1(z):=\mathcal{R}_{ \infty}(h,z,0, \dots, 0)$. But $g_\infty^1(0) \neq 0$, hence by Hurwitz' theorem, $g_\infty^1(z)$ vanishes nowhere in $D_1$. In particular, $g_\infty^1(\varepsilon_1) = \mathcal{R}_{\infty}(h,\varepsilon_1, 0,  \dots, 0) \neq 0$. 
Repeating this argument $n$ times (in the $k$-th step, the functions are $g_\Lambda^k(z)=\mathcal{R}_{\Lambda}(h, \varepsilon_1, \dots, \varepsilon_{k-1}, z,0, \dots, 0)$ and their domain of definition is $D_k = \{z \in \mathbb{C} : \abs{z} < \mathrm{Re} ( h ) - \sum_{i=1}^{k-1} \abs{\varepsilon_i} \}$), we obtain the desired result, namely that $\mathcal{R}_{ \infty}(h,\boldsymbol\varepsilon) \neq 0$. 

Having established that $\mathcal{R}_\infty$ is nowhere vanishing on the convex domain $D$, it follows that \begin{equation*}
f_\infty( h,\varepsilon_1,\dots, \varepsilon_n) := \log \big[ \mathcal{R}_\infty( h,\varepsilon_1,\dots, \varepsilon_n) \big]
\end{equation*}
is well-defined and analytic in $D$, and thus analytically continues $f_\infty$ previously defined in \eqref{TDf} on the subregion $S$ of $D$ determined by the constraints \eqref{largefieldconds}. Moreover, the functions $f_\Lambda$ defined in \eqref{fdef} converge towards $f_\infty$, uniformly on compact subsets of $D$. Similarly, any derivative of $f_\Lambda$ with respect to variables $(h, \boldsymbol{\varepsilon})$ converges towards the corresponding derivative of $f_\infty$, uniformly on compact subsets of $D$. This follows from Cauchy's integral formula for polydiscs. 

Next, we consider correlation functions. Let the (discrete) Fourier transform of spins on $\mathbb{Z}^{d}$ and the corresponding reverse transformation be denoted by
\begin{equation*} 
\hat{\sigma}_k =  \sum_{x \in \Lambda} e^{ik\cdot x} \sigma_x, \quad k \in \Lambda^*, \qquad \text{and} \qquad \sigma_x = \frac{1}{\abs{\Lambda}} \sum_{k \in \Lambda^*} e^{-ix\cdot k} \hat{\sigma}_k, \quad x \in \Lambda.
\end{equation*}
Observing that $Z_\Lambda(h, \boldsymbol{\varepsilon})$ in \eqref{modZ} can be rewritten as $Z_\Lambda(h, \boldsymbol{\varepsilon}) = Z_{\Lambda}(h) \cdot \big\langle e^{\sum_{\alpha=1}^n \varepsilon_\alpha \hat{\sigma}_{k_\alpha} } \big\rangle_{\Lambda,h}$, where $\langle \ \cdot \ \rangle_{\Lambda,h} = \langle \ \cdot \ \rangle_{\Lambda, \beta =1, h}$, it follows from definition \eqref{corrs} that
\begin{equation} \label{mastereqn}
 \frac{\partial^n \log \big[ Z_\Lambda(h, \boldsymbol{\varepsilon} ) \big]}{\partial \varepsilon_1 \cdots \partial \varepsilon_n}  \bigg|_{\varepsilon_1 = \dots = \varepsilon_n =0} =  \big\langle  \hat{\sigma}_{k_1}  ; \ \!\! \dots \  \!\! ;  \hat{\sigma}_{k_n} \big \rangle^c_{\Lambda , h}\bigg|_{\sum_{\alpha =1}^n k_\alpha = [0]}, 
\end{equation}
where the constraint $\sum_{\alpha =1}^n k_\alpha = [0]$ (with $[0] = a \mathbb{Z}^d$ for $\Lambda = \Lambda(a)=  \mathbb{Z}^d / a \mathbb{Z}^d $) follows from translation invariance, which holds because we have imposed periodic boundary conditions. With \eqref{fdef}, we thus obtain from \eqref{mastereqn}
\begin{equation} \label{mastereqn2}
 \frac{\partial^n  f_\Lambda(h, \boldsymbol{\varepsilon}  ) }{\partial \varepsilon_1 \cdots \partial \varepsilon_n}  \bigg|_{\varepsilon_1 = \dots = \varepsilon_n =0} = \frac{1}{|\Lambda|} \big\langle  \hat{\sigma}_{k_1}  ; \ \!\! \dots \  \!\! ;  \hat{\sigma}_{k_n} \big \rangle^c_{\Lambda, h} \bigg|_{\sum_{\alpha =1}^n k_\alpha = [0]}. 
\end{equation}
Proposition 7 then follows upon letting $\Lambda \nearrow \mathbb{Z}^{d}$, using that the derivative on the left-hand side has a well-defined limit \textit{analytic} in $h$ in the region $\mathbb{H}_+$, and, subsequently, Fourier-transforming back to position space. (This does not affect analyticity in $h$, because all integrations in $k$-space extend over a compact set.)
\end{proof}

Note that Proposition \ref{analytic_corrs} continues to hold for unbounded single-spin measures $\mu_0$ satisfying suitable decay assumptions at infinity (this requires a somewhat more careful analysis). In particular, this is of interest in applications to field theory (see the end of Section \ref{apps} below).

Next, we discuss generalizations of Proposition \ref{analytic_corrs} to other models satisfying a Lee-Yang theorem. For one-component spins, we consider, for example, the Hamiltonian \eqref{H2}, with a uniform external field $h$ turned on, and $\mu_0 =( \delta_{-1}+ \delta_1)/2$. We assume that the interaction $\Phi$ has spin-flip symmetry and that it is such that the Lee-Yang theorem holds at some inverse temperature $\beta >0$. Note that this requires $\beta$ to be sufficiently \textit{large}, depending on $\Phi$, if $\Phi$ is not the ferromagnetic Ising interaction; see the discussion following Theorem \ref{T3} above and references therein, in particular \cite{Ru4}. Then $\langle \sigma_{x_1}  ;  \! \! \ \dots \!\! \ ;  \sigma_{x_n} \rangle^c_{ \beta, h}$ is analytic in $h$ in the regions $\mathbb{H}_{\pm}$. The above proof is still applicable: the uniform bounds \eqref{Runifbounds} hold for general $\Phi$ with $\trinorm{\Phi} < \infty$, but the large-field cluster expansion, c.f. \eqref{TDf}, must be modified slightly. 

The analyticity results of Proposition 7  can also be extended to certain $N$-component models with $N \geq 2$. Consider the Hamiltonian \eqref{H3} and assume that condition \eqref{NBadConds} holds, that the a-priori measure $\mu_0$ satisfies the assumptions of Theorem \ref{T4} and, in addition, that $\text{supp}(\mu_0) \subset \mathbb{R}^N$ is compact. Then the corresponding partition function satisfies  
\begin{equation} \label{thecond}
Z_{ \beta, \Lambda} \big( \{ \boldsymbol{h}_x \}_{x\in \Lambda}\big)\neq 0, \qquad \text{if} \qquad \mathrm{Re} (h_x^1) > \sum_{i=2}^N |h_x^i|, \text{ for all } x \in \Lambda.
\end{equation}
We assume that $\boldsymbol{h}_x = (h,0,\dots,0)$, for all $x$ in \eqref{H3}. To conclude analyticity of $\big \langle \sigma_{x_1}^{i_1} ; \dots ;  \sigma_{x_n}^{i_n} \big \rangle_{ \beta, h}^c$ in $\mathbb{H}_+$, equation \eqref{modZ}  must be modified to read
\begin{equation*}
Z_{\beta, \Lambda} ( h,\varepsilon_1, \dots, \varepsilon_n ) := \int \exp \bigg\{ - \beta \Big[ H_\Lambda^0 - h \sum_{x \in \Lambda} \sigma_x^1 - \sum_{\alpha=1}^n \varepsilon_\alpha  \sum_{x \in \Lambda} e^{ik_\alpha \cdot x} \sigma_x^{i_\alpha} \Big] \bigg\} \prod_{x \in \Lambda} d \mu_0 (\boldsymbol{\sigma}_x). \end{equation*}
Note that the conditions in \eqref{thecond} are satisfied if $\sum_{\alpha=1}^n |\varepsilon_\alpha|< \mathrm{Re}(h)$, which is the same constraint as for $N=1$, c.f. \eqref{LYcor}. The partition function $Z_{\beta, \Lambda }(h,\boldsymbol\varepsilon)$ is thus non-vanishing and possesses an analytic logarithm on the region $D \subset \mathbb{C}^{n+1}$ defined by \eqref{domain_D}. Hence, our proof of Proposition \ref{analytic_corrs} applies in this case as well. In particular, the results hold when $\mu_0$ is the uniform measure on the unit sphere $S^{N-1}$, and, for  $N=3$ (classical Heisenberg model), \eqref{NBadConds} can be relaxed to the more natural ferromagnetic condition \eqref{NGoodConds}, by virtue of Proposition \ref{PclassHeis}.

We conclude this section by discussing joint analyticity properties of correlations in $h$ and $\beta$, which follow from the proof of 
Proposition \ref{analytic_corrs} above, using an idea of Lebowitz and Penrose; see \cite[Sections III and IV]{LePe1}. For simplicity, we restrict 
ourselves to the original framework of Proposition \ref{analytic_corrs}, but the following arguments apply to the other models discussed in the 
previous paragraph as well. We claim that (in the setting of Proposition \ref{analytic_corrs}),
\begin{equation} \label{joint_anal}
\text{$\langle \sigma_{x_1}  ;  \! \! \ \dots \!\! \ ;  \sigma_{x_n} \rangle^c_{ \beta, h} $ is jointly analytic in $\beta$ and $h$ for $\beta >0$ and $\mathrm{Re} \ \! h >0$},
\end{equation}
i.e., in some (complex) neighborhood of $(0, \infty) \times \mathbb{H}_+$ in $(\beta,h)$-space. Indeed, this can be seen as follows. Let 
$Z_{ \Lambda} ( \beta, h,\boldsymbol{\varepsilon})$ be the modified partition function $\eqref{modZ}$, but with a factor of $\beta$ inserted in the exponent, 
and the functions $f_{ \Lambda} ( \beta, h,\boldsymbol{\varepsilon})$ and $\mathcal{R}_{ \Lambda} ( \beta, h,\boldsymbol{\varepsilon})$ be formally defined 
in terms of $Z_{ \Lambda} ( \beta, h,\boldsymbol{\varepsilon})$ as in \eqref{fdef} and \eqref{defR}, respectively. For arbitrary $\beta_0 >0$, by virtue of a suitable
cluster expansion, one may select $\tilde{h}$ sufficiently large and $\delta = \delta (\beta_0, \tilde{h}) >0$ such that
\begin{equation} \label{MZ1}
\begin{split}
&\text{$\mathcal{R}_{ \infty} ( \beta, h,\boldsymbol{\varepsilon}) := \lim_{\Lambda \nearrow \mathbb{Z}^d} \mathcal{R}_{ \Lambda} ( \beta, h,\boldsymbol{\varepsilon})$
is jointly analytic in the} \\
&\text{variables $(\beta,h,\boldsymbol{\varepsilon})$ on the polydisc ${D}_{\beta_0}(\delta)\times{D}_{\tilde{h}}(\delta) \times \big({D}_0(\delta) \big)^n$,}
\end{split}
\end{equation}
where $D_z(r)= \{ z' \in \mathbb{C}: |z'-z| < r\}$ denotes the open disk of radius $r$ centered at $z$. For arbitrary $h_0 \in \mathbb{H}_+$, let $\overline{K} \subset \mathbb{H_+}$ 
be a closed rectangle containing both $h_0$ and $D_{\tilde{h}}(\delta)$ in its interior, and let $I = D_{\beta_0} ( \delta) \cap \mathbb{R}_+$ be the (real) diameter of the disk
$D_{\beta_0} ( \delta)$. The proof of Proposition \ref{analytic_corrs} yields that
\begin{equation} \label{MZ2}
\text{$\mathcal{R}_{ \infty} ( \beta, h,\boldsymbol{\varepsilon})$ is analytic in $(h,\boldsymbol{\varepsilon})$ on $K\times \big( D_0(\delta) \big)^n$, for every $\beta \in I$},
\end{equation}
provided $\delta >0$ is sufficiently small. Moreover, it follows from \eqref{Runifbounds} that
\begin{equation} \label{MZ3}
\sup_{ \beta \in I \ \!\! ; \ \!\! h \in K \ \!\! ; \ \!\! \varepsilon_\alpha \in  D_0(\delta)} \abs{\mathcal{R}_{ \infty} (\beta,h,\boldsymbol{\varepsilon}  )} < \infty.
\end{equation}
Together with \eqref{MZ1}, \eqref{MZ2} and \eqref{MZ3}, the Malgrange-Zerner theorem \cite[Theorem 2.2]{Sp1} (see also the lemma in Section III of \cite{LePe1}) then 
implies that $\mathcal{R}_{ \infty} ( \beta, h,\boldsymbol{\varepsilon})$ is jointly analytic for $\beta \in I$, $h \in K$, and 
$\varepsilon_\alpha \in D_0(\delta)$, $\alpha=1,\dots,n$ (i.e., in some neighborhood of this set in $\mathbb{C}^{2+n}$). In particular, 
$\mathcal{R}_{ \infty} ( \beta, h,\boldsymbol{\varepsilon})$ is jointly analytic in a polydisc around $(\beta_0,h_0,\boldsymbol{0})$, and we may further assume that 
$\mathcal{R}_{ \infty}$ is nowhere vanishing within this polydisc, by continuity (since $\mathcal{R}_{ \infty}(\beta_0,h_0,\boldsymbol{0}) \neq 0$, 
as shown in the proof of Proposition \ref{analytic_corrs}). Thus, $f_\infty = \log \mathcal{R}_{ \infty}$ is analytic on this polydisc, and so is
\begin{equation*}
\lim_{\Lambda \nearrow \mathbb{Z}^d} \frac{1}{|\Lambda|} \big\langle  \hat{\sigma}_{k_1}  ; \ \!\! \dots \  \!\! ;  \hat{\sigma}_{k_n} \big \rangle^c_{\Lambda, \beta, h} \bigg|_{\sum_{\alpha =1}^n k_\alpha = [0]},
\end{equation*}
by \eqref{mastereqn2}. Since $\beta_0 > 0$ and $h_0 \in \mathbb{H_+}$ were arbitrary, \eqref{joint_anal} now follows as before upon Fourier-transforming back to position space, which does not affect analyticity 
in $\beta$ and $h$.

\section{Applications} \label{apps}

In this section, we discuss various applications of Proposition \ref{analytic_corrs} and of Eq. (\ref{mastereqn2}). 

Our first application concerns the independence of our analyticity results of the choice of boundary conditions; (recall that, in the above proof, we have imposed \textit{periodic} boundary conditions). We wish to show that Proposition \ref{analytic_corrs} continues to hold for \textit{all} boundary conditions, $b$, for which the Lee-Yang theorem holds -- this includes, in particular, free boundary conditions -- and that the thermodynamic limit of correlations is unique for this class of boundary conditions, provided $h$ belongs to the Lee-Yang region $\mathbb{H}_+$. 

For the sake of simplicity, we sketch our arguments in the setting of Proposition \ref{analytic_corrs}, but the following conclusions continue to hold for any of the Lee-Yang models $(\Phi, \mu_0)$ mentioned above for which (an analogue of) Proposition \ref{analytic_corrs} has been shown to hold. Thus, let $H_\Lambda$ be the Ising Hamiltonian \eqref{His} with $h_x=h$, for all $x$, and pair couplings $J_{xy}$ satisfying \eqref{Jconds}, and let $\mu_0$ be an even measure with compact support satisfying condition \eqref{zerointsLYcond}.  Using a cluster expansion of the correlations at large magnetic fields, one proves (see for example \cite[Theorem V.7.11]{Si1}) that the model has a unique equilibrium state in a region $\Omega_\beta \subset \mathbb{H}_+$ defined by the constraints
\begin{equation*}
\mathrm{Re}(h) > h_0(\beta), \qquad \abs{\mathrm{Im}(h)} < \varepsilon (\beta),
\end{equation*}  
for some large $h_0 >0$ and some $\varepsilon > 0$ (both depending on $\beta$). Denoting by $\langle \ \cdot \ \rangle_{\Lambda, \beta, h, b}$ the (finite-volume) thermal average corresponding to a boundary condition $b$, it follows that 
\begin{equation}\label{bc}
\lim_{\Lambda \nearrow \mathbb{Z}^{d}} \frac {1}{\vert\Lambda\vert} \langle \hat{\sigma}_{k_1}  ;   \! \! \ \dots  ;  \! \! \ \hat{\sigma}_{k_n} \rangle^c_{\Lambda, \beta, h,b}
\end{equation}
exists and is independent of $b$, for any $h \in \Omega_\beta$. 
In particular, the limit agrees with
\begin{equation}\label{per.bc}
\lim_{\Lambda\nearrow\mathbb{Z}^d} \frac{1}{\vert\Lambda\vert} \langle \hat{\sigma}_{k_1}  ;  \! \! \ \dots \!\! \ ;  \hat{\sigma}_{k_n} \rangle^c_{\Lambda, \beta, h},
\end{equation}
for $h \in \Omega_\beta$, where, in the latter correlations, periodic boundary conditions are imposed, as in Section 3. Note that the correlation in (\ref{per.bc}) vanishes unless $\sum_{\alpha = 1}^{n} k_{\alpha} = [0]$ and is finite if the latter condition holds. Assuming that the boundary condition $b$ is such that the Lee-Yang theorem holds for $h\in\mathbb{H}_{+}$, we may use Eqs.
(\ref{mastereqn}) and (\ref{mastereqn2}) and then apply the Lee-Yang theorem to the partition function and the free energy of the model with boundary condition $b$ imposed at $\partial \Lambda$ to show that the correlations  
$\frac {1}{\vert\Lambda\vert} \langle \hat{\sigma}_{k_1}  ;   \! \! \ \dots  ;  \! \! \ \hat{\sigma}_{k_n} \rangle^c_{\Lambda, \beta, h,b}$
are uniformly bounded and analytic in $h$ on $\mathbb{H}_+$. 
Thus, multiplying 
$\frac {1}{\vert\Lambda\vert} \langle \hat{\sigma}_{k_1}  ;   \! \! \ \dots  ;  \! \! \ \hat{\sigma}_{k_n} \rangle^c_{\Lambda, \beta, h,b}$ by $\exp(-i\sum_{\alpha = 1}^{n} k_{\alpha}\cdot x_{\alpha})$ and integrating over the surface defined by the equation $\sum_{\alpha =1}^{n} k_{\alpha} = [0]$, we find, using that the integration domain is compact, that 
\begin{equation} \label{bTDlimit}
 \lim_{\Lambda \nearrow \mathbb{Z}^{d}} \langle \sigma_{x_1}  ;   \! \! \ \dots  ;  \! \! \ \sigma_{x_n} \rangle^c_{\Lambda, \beta, h,b} = \langle \sigma_{x_1}  ;  \! \! \ \dots \!\! \ ;  \sigma_{x_n} \rangle^c_{ \beta, h}, \qquad \text{for all } h \in \mathbb{H}_+,
\end{equation}
which proves analyticity of the thermodynamic limit of correlations with boundary condition $b$ on the \textit{entire} half-plane $\mathbb{H}_+$. 

\bigskip

\noindent Next, we consider the magnetization $\langle \sigma_x \rangle_{\beta,h}$, which, by Proposition \ref{analytic_corrs}, is an analytic function of $h$ in $ \mathbb{H}_+$, for any $\beta > 0$. Since $\beta$ is fixed in the sequel, we omit it from our notation. We propose to show that 
\begin{equation} \label{magnetization_properties}
\text{$\langle \sigma_x \rangle_{h}$ is a strictly positive, increasing, concave function of $h >0$.}
\end{equation}
Note that this yields a well-known bound on a critical exponent for the magnetization as a function of $h$, (with $\beta = \beta_{c}$, the critical inverse temperature).
As a preliminary step, we show that $\langle \sigma_{0}  ;  \sigma_{x} \rangle^c_{h}$ is decreasing for $h>0$. We recall that
\begin{equation} \label{useful_identity}
\frac{\partial}{\partial h} \langle \sigma_{x_1}  ;  \! \! \ \dots \!\! \ ;  \sigma_{x_n} \rangle^c_{\Lambda, h} = \sum_{z \in \Lambda}\langle \sigma_{x_1}  ;  \! \! \ \dots \!\! \ ;  \sigma_{x_n}; \sigma_{z} \rangle^c_{\Lambda, h},
\end{equation}
Applying this identity for $n=2$ and using the GHS-inequality (see  \cite{GHS1} for Ising spins, \cite{SiGr1} for more general $\mu_0$), we see that 
\begin{equation*}
\frac{\partial}{\partial h} \langle \sigma_{0}  ;  \sigma_{x} \rangle^c_{\Lambda, h} = \sum_{z \in \Lambda} \langle \sigma_{0}  ;  \sigma_{x}; \sigma_z \rangle^c_{\Lambda, h} \leq 0,
\end{equation*}
for all $h \geq 0$. Letting $\Lambda \nearrow \mathbb{Z}^d$, one obtains that $\frac{\partial}{\partial h} \langle \sigma_{0}  ;  \sigma_{x} \rangle^c_{ h} \leq 0$, for all $h >0$. Indeed, this follows from Cauchy's integral formula, using that $ \langle \sigma_{0}  ;  \sigma_{x}\rangle^c_{\Lambda, h}$ tends towards its infinite-volume limit uniformly on compact subsets of $\mathbb{H_+}$. (The latter claim follows from the proof of Proposition \ref{analytic_corrs}: invoking Vitali's theorem, it suffices to establish uniform bounds, $\sup_{\Lambda; h} \abs{ \langle \sigma_{0}  ;  \sigma_{x}\rangle^c_{\Lambda, h}} < \infty$, for $h$ belonging to an arbitrary compact subset $K$ of $ \mathbb{H}_+$, which, in turn, follow immediately from \eqref{mastereqn2}, \eqref{Runifbounds}, and the fact that  $\langle \sigma_{0}  ;  \sigma_{x} \rangle^c_{\Lambda, h} = |\Lambda|^{-2} \sum_{k \in \Lambda^*} e^{ik \cdot x}  \langle \hat{\sigma}_{k}  ;   \hat{\sigma}_{-k}  \rangle^c_{\Lambda, h}$.) 

Returning to \eqref{magnetization_properties}, the identity \eqref{useful_identity} and the GHS-inequality imply that $\partial^2 \langle \sigma_x \rangle_{h} / \partial h^2 \leq 0$, for $h >0$. Similarly, monotonicity of $\langle \sigma_x \rangle_{h}$, for $h >0$, follows from \eqref{useful_identity} and the FKG-inequality (which holds since the interactions are ferromagnetic).  It remains to show that $\langle \sigma_{x} \rangle_{ h}$ is positive for $h >0$. Indeed, $\langle \sigma_{x} \rangle_{ h} \geq 0$, for all $h>0$, follows from Griffiths' inequality. If $\langle \sigma_{x} \rangle_{ h}$ vanished, for some $h>0$, then $\langle \sigma_{x} \rangle_{ h'} = 0$, for all $0 < h' \leq h$, by monotonicity, which is impossible, because the zeroes of $\langle \sigma_{x} \rangle_{ h}$ form a discrete subset of $\mathbb{H}_+$. Similarly, the first derivative of $\langle \sigma_{x}\rangle_{h}$ in $h$ is strictly positive and the second derivative has, at most, a discrete set of zeros, for $h>0$. (Note that, since $\langle \sigma_{x} \rangle_{ -h}= - \langle \sigma_{x} \rangle_{ h}$, one derives corresponding properties of the function $\langle \sigma_{x} \rangle_{ h}$, $h < 0$, from \eqref{magnetization_properties}.)

\bigskip

\noindent Next, we consider the \textit{mass gap}, $m(\beta,h)$ (inverse correlation length), defined as 
\begin{equation} \label{mass_gap}
 m (\beta,h) = \frac{1}{\xi(\beta,h)} =- \limsup_{|x| \to \infty} \frac{1}{\abs{x}} \log\big| \langle \sigma_{0}  ;  \sigma_{x} \rangle^c_{\beta, h} \big|.
 \end{equation}
We restrict our attention to the generalized Ising model considered in Proposition \ref{analytic_corrs} with \textit{finite-range} interactions (i.e., $J_{xy} = 0$ whenever $|x-y| \geq R$, for some $R \geq 1$), but the following discussion applies to more general spin models obeying a Lee-Yang Theorem and to Euclidean $\lambda \phi^{4}$-field theories with non-zero ``external field'' $h$, in two and three space-time dimensions; see \cite{LePe1}, \cite{LePe2}, and \cite{GRS}. The mass gap satisfies
\begin{equation} \label{positivity_mass_gap}
 m (\beta,h) > 0, \qquad \text{for all $h \in \mathbb{H}_+$ and $ \beta >0$},
\end{equation}
i.e., the two-point function $ \langle \sigma_{0}  ;  \sigma_{x} \rangle^c_{\beta, h}$ exhibits exponential clustering. The inequality \eqref{positivity_mass_gap} can be shown as follows. By Proposition \ref{analytic_corrs}, $\langle \sigma_{0}  ;  \sigma_{x} \rangle^c_{h}$ is analytic in $ \mathbb{H}_+$, hence $ \log \big| \langle \sigma_{0}  ;  \sigma_{x} \rangle^c_{h} \big|$ is a $subharmonic$ function of $h \in \mathbb{H}_+$, see \cite[Theorem A.3]{GRS}. Using a cluster expansion, one shows that $m$ is positive on an open subset of $\mathbb{H_+}$ corresponding to sufficiently large $\mathrm{Re} \ \!  h$; (actually, positivity on some smooth arc in $\mathbb{H}_+$ would suffice). It then follows from subharmonicity (see \cite[Lemmas 1 and 5]{LePe2}) that $m$ is positive everywhere on $\mathbb{H_+}$.

Next, we remark that
\begin{equation}\label{mass_gap_incr}
 \text{$m(\beta,h)$ is an increasing function of $h>0$, for all $\beta>0$.}
 \end{equation}
This follows from definition \eqref{mass_gap}, using monotonicity of the logarithm.  For, we have already shown that the two-point function $\langle \sigma_{0}  ;  \sigma_{x} \rangle^c_{\beta, h}$ is decreasing in $h$, for $h>0$. 
We also wish to recall a bound on the critical exponent $\delta$ describing the divergence of the correlation length $\xi = m^{-1}$, as $h \searrow 0$, at the critical inverse temperature $\beta_c$; (see \cite{GRS} for a more general result that extends to Euclidean $\lambda \phi^{4}$-field theory):
 \begin{equation} \label{crit_exp}
 \xi(\beta_{c},h) \sim h^{-\delta}, \ \text{with } \delta \leq 1.
 \end{equation}
This follows by showing that $m(\beta_c, h) \geq c \cdot h$, for some positive constant $c$ and all sufficiently small $h>0$. The latter is a consequence of Theorem A.6 in \cite{GRS}, using the fact that the functions $-  \frac{1}{\abs{x}} \log\big| \langle \sigma_{0}  ;  \sigma_{x} \rangle^c_{\beta_c, h} \big|$ are superharmonic in $\mathbb{H}_+$ and that, given any (real) $h_0>0$, these functions are bounded away from $0$ by a positive constant (uniform in $h$ and $x$), for all $h \geq h_0$ and all sufficiently large $x$, which follows from \eqref{positivity_mass_gap} and \eqref{mass_gap_incr}.

\bigskip

Finally, we mention a generalization (alluded to, above) of Proposition \ref{analytic_corrs} to some $N$-component ($N=1,2,3$) Euclidean $\lambda \vert\boldsymbol{\phi}\vert^{4}_d$-field theories ($\boldsymbol{\phi} = (\phi^1, \cdots, \phi^N)$) in $d=2,3$ space-time dimensions, with periodic boundary conditions. The correlation functions of lattice spin systems are replaced by the Schwinger (Euclidean Green) functions,
\begin{equation*}
S_{L,a,h} (x_1,\alpha_1, \dots, x_n, \alpha_n), 
\end{equation*}
 of the properly renormalized lattice field theory on a lattice $\mathbb{Z}_{a}^{d}\cap \Lambda_{L}$, with periodic boundary conditions imposed at $\partial \Lambda_L$, where $d=2,3$, $a>0$ is the lattice spacing, $\Lambda_{L} = [-L/2,L/2]^d$ is a cube in $\mathbb{Z}^{d}_{a}$ with sides containing $\frac{L}{a}$ sites, and the arguments $x_i,\alpha_i$, $i=1,\dots,n,$ stand for the field components $\phi^{\alpha_i}(x_i)$. 
Formally, the Schwinger functions are given by
 \begin{equation*}
S_{L,a, h} (x_1, \alpha_1,\dots, x_n,\alpha_n) = 
\frac{\int \phi^{\alpha_1(x_1)} \cdots \phi^{\alpha_n}(x_n)e^{-A_{L,a,h}(\boldsymbol{\phi})} D\boldsymbol{\phi}_{L,a}}{  \int e^{-A_{L,a,h}(\boldsymbol{\phi})} D\boldsymbol{\phi}_{L,a}}, 
\end{equation*}
where $A_{L,a,h}$ is the Euclidean action of the theory, which is identical to the Hamilton function with periodic boundary conditions of the corresponding classical lattice spin system (with nearest-neighbor couplings, and in an external magnetic field $h$), but with coupling constants that depend on the lattice spacing $a$ in such a way that the continuum limit $a \rightarrow 0$ exists; see, e.g., \cite{Si3}, \cite{Park}.

Combining results in \cite{Si3} and \cite{Sp1}, for $d=2$, and in \cite{Park} and \cite{Feld-Ost}, \cite{Sp1}, for $d=3$, one can prove an analogue of Proposition \ref{analytic_corrs}, for $h\in\mathbb{H}_+$. This is accomplished by first proving existence and Euclidean invariance of the limits
\begin{equation}\label{limits}
\lim_{L\nearrow \infty}  \lim_{a\searrow 0} S_{L,a, h} (x_1, \alpha_1,\dots, x_n,\alpha_n),
\end{equation}
for $\mathrm{Re} \ \! h$ large enough and $\mathrm{Im} \ \! h$ small enough. The Lee-Yang theorem (see \cite[Theorem 6]{SiGr1}) and uniform bounds on the analogue of the free energy, see (\ref{fdef}), discussed in \cite{Si3} ($d=2$) and in \cite{GlJa} ($d=3$) then imply analogues of Eqs. (\ref{Runifbounds}) and (\ref{mastereqn2}) that can be used to prove bounds on the Schwinger functions, integrated against test functions on momentum space, that are uniform in $a$ and $L$ and yield analyticity in $h$ on $\mathbb{H}_+$. As a consequence, the limiting Schwinger functions in (\ref{limits}) exist, are Euclidean invariant and analytic in $h$, for $h\in\mathbb{H}_+$. A more detailed discussion of these arguments goes beyond the scope of this note.

The results mentioned here are of interest in an analysis of phase transitions accompanied by spontaneous symmetry breaking in $\lambda \vert \boldsymbol{\phi}\vert ^{4}-$ theory in $d=3$ space-time dimensions; see \cite[Section 4]{FSS}.\\

\medskip

\noindent \textit{Acknowledgements}. We thank David Ruelle for informing us about ref. \cite{LeRu1} prior to publication and for several helpful discussions. We are grateful to Barry Simon for helping us to trace some references. The senior author thanks Elliott Lieb for all he has taught him and for his friendship.

\end{document}